\documentclass[11pt]{article}
\usepackage{a4wide}

\usepackage{amsmath}
\usepackage{amssymb}
\usepackage{graphicx}
\usepackage{url}
\usepackage{comment}
\usepackage{enumitem}

\usepackage{amsfonts}
\usepackage{amsthm}
\usepackage{comment}
\usepackage{thmtools}
\usepackage{thm-restate}
\usepackage{mdframed}
\usepackage{booktabs}

\newtheorem{theorem}{Theorem}[]
\newtheorem{lemma}[theorem]{Lemma}
\newtheorem{corollary}[theorem]{Corollary}
\newtheorem{remark}[theorem]{Remark}
\newtheorem{proposition}[theorem]{Proposition}

\newtheorem{conjecture}[theorem]{Conjecture}
\newtheorem{problem}[theorem]{Problem}

\usepackage{todonotes}

\usepackage{xcolor}
\usepackage{colortbl}
\colorlet{tableheadcolor}{gray!25} 
\colorlet{tablerowcolor}{gray!20} 
\newcommand{\rowcol}{\rowcolor{tablerowcolor}} %

\usepackage{graphicx}
\usepackage{booktabs}
\usepackage{pdflscape}
\usepackage{mdframed}
\usepackage{subfig}
\usepackage{mathtools}

\newcommand{\NP}{\textsc{NP}}

\usepackage{xspace}
\usepackage{framed}

\newcommand{\ProblemFormat}[1]{{\sc #1}}
\newcommand{\ProblemName}[1]{\ProblemFormat{#1}\xspace}

\DeclareMathOperator{\rom}{\gamma_R}
\DeclareMathOperator{\romtwo}{\gamma_{\{R2\}}}
\DeclareMathOperator{\pid}{\gamma^p_I}
\DeclareMathOperator{\KC}{KC}
\DeclareMathOperator{\im}{\nu_s}
\DeclareMathOperator{\fdk}{fd_k}
\DeclareMathOperator{\fd2}{fd_2}

\newcommand{\probXC}{\ProblemName{Exact Cover by 3-Sets}}
\newcommand{\probPlanarXC}{\ProblemName{Planar Exact Cover by 3-Sets}}
\newcommand{\probRomanTwoD}{\ProblemName{Roman $\{2\}$-Domination}}
\newcommand{\probPID}{\ProblemName{Perfect Italian Domination}}

\makeatletter
\g@addto@macro\@floatboxreset\centering
\makeatother

\newcount\bsubfloatcount
\newtoks\bsubfloattoks
\newdimen\bsubfloatht

\makeatletter
\newcommand{\bsubfloat}[2][]{%
  \sbox\z@{#2}%
  \ifdim\bsubfloatht<\ht\z@
    \bsubfloatht=\ht\z@
  \fi
  \advance\bsubfloatcount\@ne
  \@namedef{bsubfloat\romannumeral\bsubfloatcount}{%
    \subfloat[#1]{\vbox to\bsubfloatht{\hbox{#2}\vfill}}}%
}
\newcommand{\resetbsubfloat}{\bsubfloatcount\z@\bsubfloatht=\z@}
\makeatother

\title{Perfect Italian domination on planar and regular graphs}

\author{Juho Lauri \and Christodoulos Mitillos\thanks{University of Cyprus, Cyprus}}

\begin{document}
\maketitle

\begin{abstract}
A perfect Italian dominating function of a graph $G=(V,E)$ is a function $f : V \to \{0,1,2\}$ such that for every vertex $f(v) = 0$, it holds that $\sum_{u \in N(v)} f(u) = 2$, i.e., the weight of the labels assigned by $f$ to the neighbors of $v$ is exactly 2.
The weight of a perfect Italian function is the sum of the weights of the vertices.
The perfect Italian domination number of $G$, denoted by $\pid(G)$, is the minimum weight of any perfect Italian dominating function of $G$.
While introducing the parameter, Haynes and Henning (Discrete Appl.\ Math.\ (2019), 164--177) also proposed the problem of determining the best possible constants $c_\mathcal{G}$ such that $\pid(G) \leq c_\mathcal{G} \times n$ for all graphs of order $n$ when $G$ is in a particular class $\mathcal{G}$ of graphs.
They proved that $c_\mathcal{G} = 1$ when $\mathcal{G}$ is the class of bipartite graphs, and raised the question for planar graphs and regular graphs.
We settle their question precisely for planar graphs by proving that $c_\mathcal{G} = 1$ and for cubic graphs by proving that $c_\mathcal{G} = 2/3$.
For split graphs, we also show that $c_\mathcal{G} = 1$.
In addition, we characterize the graphs $G$ with $\pid(G)$ equal to~2 and~3 and determine the exact value of the parameter for several simple structured graphs.
We conclude by proving that it is $\NP$-complete to decide whether a given bipartite planar graph admits a perfect Italian dominating function of weight~$k$.
\end{abstract}

\section{Introduction}

The motivation for the problem we study stems from the problem of deployment of military forces to guard several points of interest, modeled by an undirected graph.
Such problems from different historical eras were described by ReVelle and Rosing~\cite{Revelle2000} (see also Stewart~\cite{Stewart1999}). 
For instance, the authors describe a defense-in-depth strategy by Emperor Constantine (Constantine the Great, 274--337) where units were deployed such that any city without a unit was to be neighbored by a city harboring two units.
The idea was that if the city without a unit was attacked, the neighboring city could dispatch a unit to protect it without becoming vulnerable itself.
In this setting, the objective was to minimize the total number of units needed.
Albeit overly simplified particularly for the modern era to be of practical use, these type of domination problems on graphs have resulted in interesting graph-theoretical problems that have attracted significant interest from the research community.

Let $G=(V,E)$ be a simple undirected graph.
To reduce clutter, we can write an element $\{u,v\} \in E$ as $uv$.
The \emph{open neighborhood} of a vertex $v \in V$, denoted by $N(v)$, is the set of neighbors of $v$ excluding $v$ itself, i.e., $N(v) = \{ u \mid uv \in E \}$.
The \emph{degree} of a vertex $v$ is the number of edges incident to it, i.e., $|N(v)|$.
In particular, a vertex of degree one is a \emph{pendant} and the vertex adjacent to a pendant vertex is a \emph{support}.
For the following discussion, let $f : V \to \{0,1,2\}$ be a vertex-labeling of~$G$.

We say that $f$ is a \emph{perfect Italian dominating function} on $G$, abbreviated a PID-function, when it holds that whenever $f(v) = 0$ for any $v \in V$, it holds that $\sum_{u \in N(v)} f(u) = 2$, i.e., the accumulated weight assigned to the neighbors of $v$ by $f$ is exactly 2.
The \emph{weight} of $f$ is the sum of its labels, i.e., $\sum_{v \in V} f(v)$.
The \emph{perfect Italian domination number} of $G$, denoted by $\pid(G)$, is the minimum weight of a PID-function on $G$.
This concept was introduced by Haynes and Henning~\cite{Haynes2019} as a natural variant of similar, previously rather heavily-studied, parameters of so-called Roman domination introduced by Cockayne~{et al.}~\cite{Cockayne2004}. We refer the interested reader to e.g.,~\cite[Section~3.9]{Gera2018} for a brief overview of some of these variants, but describe some relevant to our work in the following.

We say that $f$ is a \emph{Roman dominating function}, abbreviated an RDF-function, on $G$ if every vertex $v \in V$ for which $f(v) = 0$ is adjacent to at least one vertex $u$ for which $f(u) = 2$.
The \emph{Roman domination number} of $G$, denoted by $\rom(G)$, is the minimum weight of an RDF-function on~$G$.
While introducing the concept, Cockayne~{et al.}~\cite{Cockayne2004} also gave several bounds for $\rom(G)$ and determined its value for certain structured graph classes including paths, cycles and complete multipartite graphs.
For example, the authors proved that $\gamma(G) \leq \rom(G) \leq 2 \gamma(G)$ and that $\gamma(G) = \rom(G)$ implies $G$ to be edgeless, where $\gamma(G)$ is the \emph{domination number} of~$G$.
Further, they mentioned that it has been proved that deciding whether a graph $G$ admits an RDF-function of weight at most~$k$ is $\NP$-complete.
For further combinatorial results on $\rom(G)$, see the survey~\cite[Section~5.7]{Chellali2012}.
A possible application in network design is described by Chambers~{et al.}~\cite{Chambers2009}, while Liedloff~{et al.}~\cite{Liedloff2008} give algorithms for several structured graph classes.

Another variant of perfect Italian domination, introduced by Chellali~{et al.}~\cite{Chellali2016}, is obtained by relaxing the constraint so that for every $v \in V$, if $f(v) = 0$, then $\sum_{u \in N(v)} f(u) \geq 2$, i.e., the accumulated weight of $f$ assigned to the neighbors of $v$ is \emph{at least} two.
Such an $f$ is known as a \emph{Roman $\{2\}$-dominating function of $G$}, also referred to as an Italian dominating function by Henning and Klostermeyer~\cite{Henning2017}.
Here, the \emph{Roman $\{2\}$-domination number} of $G$, denoted by $\romtwo(G)$, is the minimum weight of a Roman $\{2\}$-dominating function on $G$.
In addition to various combinatorial results, Chellali~{et al.}~\cite{Chellali2016} also proved that deciding whether a graph $G$ admits a Roman $\{2\}$-dominating function of weight at most~$k$ is $\NP$-complete even when $G$ is bipartite.

\paragraph{Our results}
We continue the study of perfect Italian domination initiated by Haynes and Henning~\cite{Haynes2019} by giving the following results.
\begin{itemize}
\item In Section~\ref{sec:basic}, we relate the perfect Italian domination number to other well-known Roman domination numbers. 
Further, we characterize the graphs $G$ such that $\pid(G) = 2$ which includes connected threshold graphs, paths, cycles, and wheels.
We proceed to give a characterization of graphs $G$ such that $\pid(G) = 3$, and then conclude by determining the exact value of the parameter for complete multipartite graphs.

\item In Section~\ref{sec:upper-bounds}, we consider the question of Haynes and Henning~\cite{Haynes2019} for finding best possible upper bounds on $\pid(G)$ as a function of the order $n$ when $G$ is planar or regular.
For planar graphs, split graphs, and $k$-regular graphs for $k \geq 5$, we prove that there is an infinite family of such graphs $G$ such that $\pid(G) = n$, meaning that no upper bound of the form $c \cdot n$ exists, for any $c < 1$.
For cubic graphs, we prove that $\tfrac{2}{5}n \leq \pid(G) \leq \tfrac{2}{3}n$, and demonstrate that these bounds are tight.

\item In Section~\ref{sec:hardness}, we turn to complexity-theoretic questions.
Specifically, we prove that deciding whether a given graph $G$ admits a PID-function of weight at most~$k$ is $\NP$-complete, even when $G$ is restricted to the class of bipartite planar graphs.
We also strengthen the result of Chellali~{et al.}~\cite{Chellali2016} by showing that deciding whether $G$ admits a Roman $\{2\}$-dominating function of weight at most~$k$ is $\NP$-complete, even when $G$ is both bipartite and planar.
\end{itemize}
We conclude in Section~\ref{sec:problems} by giving some further open problems and conjectures arising from our work.

\section{Basic bounds, properties and characterizations}
\label{sec:basic}
In this section, we determine some basic properties of the perfect Italian domination number of a graph.

\subsection{Graphs with perfect Italian domination number two}
We begin with the following known bounds.
\begin{theorem}[Chellali~{et al.}~\cite{Chellali2016}]
\label{thm:bounds}
For every graph $G$, it holds that $\gamma(G) \leq \romtwo(G) \leq \rom(G)$.
\end{theorem}

\begin{proposition}
\label{prop:pid-bounds}
For every graph $G$, it holds that $\gamma(G) \leq \romtwo(G) \leq \pid(G)$.
\end{proposition}
\begin{proof}
Every PID-function of $G$ is a Roman $\{2\}$-dominating function of $G$, so the bound follows.
\end{proof}
\noindent Since the components of a graph do not interact with each other in terms of domination, the optimal PID-function of a graph $G$ consists of optimal PID-functions of its components, as made precise in the following.
\begin{proposition}
\label{prop:disconnected}
If $G$ is a disconnected graph with components $G_1,G_2,\ldots,G_r$, then $\pid(G) = \sum_{i=1}^{r} \pid(G_i)$.
\end{proposition}
\noindent It was shown by Chellali~{et al.}~\cite[Corollary~9]{Chellali2016} that for $n$-vertex paths $P_n$ and cycles $C_n$ there is an optimal Roman $\{2\}$-dominating function that uses only weights of 0 and 1. Such a function is also a PID-function since the graphs are of maximum degree two, meaning that any vertex of weight 0 has to have exactly two neighbors of weight 1. Furthermore, if in a given PID-function a vertex of weight 0 has a neighbor of weight 2, we end up with the pattern 2-0-0-2, which is no better than the above. Combining these two points, we arrive at the following.
\begin{proposition}
For every integer $n \geq 1$, it holds that $\pid(P_n) = \lceil (n+1)/2 \rceil$ and $\pid(C_n) = \lceil n/2 \rceil$.
\end{proposition}
\noindent The following observation characterizes the graphs $G$ with $\pid(G) = 2$.
Recall that the \emph{join} of graphs $G$ and $H$ is the graph union of $G$ and $H$ with all the edges between $V(G)$ and $V(H)$ added.
\begin{proposition}
\label{prop:dom}
A non-trivial connected graph $G$ has $\pid(G) = 2$ precisely when $G$ can be written as the join of $G_1$ and $G_2$, where $G_1$ is either $K_1$, $2K_1$ or $K_2$.
\end{proposition}
\begin{proof}
For $G$ to have $\pid(G) = 2$, there must exist a PID-function that labels (i) exactly one vertex 2 and the rest 0 or (ii) exactly two vertices 1 and the rest 0.
If exactly one vertex $v$ has label 2, all vertices distinct from $v$ must be adjacent to it, i.e., $G_1$ must be $K_1$.
Similarly, if there are two vertices $u$ and $v$ with label 1, $G_1$ must be either $2K_1$ or $K_2$ meaning that $u$ dominates at least $V(G) \setminus \{v\}$ and vice versa for $v$.
\end{proof}
\noindent Several structured graph classes fall under the above characterization, as we will see next.
\begin{proposition}
\label{prop:pid-threshold}
A non-trivial connected threshold graph $G$ has $\pid(G) = 2$.
\end{proposition}
\begin{proof}
Every threshold graph $G$ can be represented as a binary string $s(G)$, read from left to right, where 0 denotes the addition of an isolated vertex and 1 denotes the addition of a dominating vertex (for a proof, see~\cite[Theorem~1.2.4]{Mahadev1995}).
Because $G$ is connected, the last symbol of $s(G)$ is a 1.
As $G$ has a dominating vertex, the proof follows by Proposition~\ref{prop:dom}.
\end{proof}
\noindent The following results are now immediate, where $S_n$, $K_n$, and $W_n$ denote the star graph, complete graph, and wheel graph, respectively, on $n$ vertices.
\begin{proposition}
For every integer $n \geq 2$, it holds that $\pid(S_n) = 2$.
\end{proposition}

\begin{proposition}
For every integer $n \geq 2$, it holds that $\pid(K_n) = 2$.
\end{proposition}

\begin{proposition}
For every integer $n \geq 4$, it holds that $\pid(W_n) = 2$.
\end{proposition}

All such graphs have a dominating vertex, so the result follows. We close with one additional consequence of Proposition~\ref{prop:dom}.

\begin{proposition}
For every integer $n \geq 1$, it holds that $\pid(K_{2,n}) = 2$. 
\end{proposition}
\begin{proof}
The graph $K_{2,n}$ can be written as the join of $2K_1$ and $\overline{K}_n$ (i.e., the edgeless $n$-vertex graph) yielding this result.
\end{proof}

\subsection{Bounds via fair domination}
In this subsection, we give a characterization of graphs $G$ with $\pid(G) = 3$.
In order to do so, let us first introduce some concepts from domination.

Let $G=(V,E)$ be a graph.
For $k \geq 1$, a \emph{$k$-fair dominating set} of $G$ is a dominating set $D$ such that $|N(v) \cap D| = k$ for every $v \in V \setminus D$. 
That is, every vertex not in $D$ has precisely $k$ neighbors in $D$.
The \emph{$k$-fair domination number} of $G$, denoted by $\fdk(G)$, is the minimum cardinality of a $k$-fair dominating set in $G$.
This concept was introduced by Caro~{et al.}~\cite{Caro2012} (see also~\cite{Hansberg2013}).
It is also captured by the concept of $[j,k]$-domination as introduced by Chellali~{et al.}~\cite{Chellali2013}.
Here, a subset $S \subseteq V$ is a $[j,k]$-set if for every vertex $v \in V \setminus S$ it holds that $j \leq |N(v) \cap S| \leq k$, that is, every vertex not in $S$ has at least $j$ but no more than $k$ neighbors in $S$.
Clearly, a $k$-fair dominating set is equivalent to a $[k,k]$-dominating set.
Finally, such a set is also known as a \emph{perfect $k$-dominating set} (see e.g.,~\cite{Chaluvaraju2010,Chaluvaraju2014}).

\begin{theorem}
For every graph $G$, it holds that $\pid(G) \leq \fd2(G)$.
\end{theorem}
\begin{proof}
Let $D$ be a 2-fair dominating set.
Construct a vertex-labeling $f$ such that $f(v) = 1$ for $v \in D$ and $f(u) = 0$ for $u \notin D$.
By definition, for every $u$ with $f(u) = 0$ there are precisely two vertices $v_1$ and $v_2$ with $f(v_1) = f(v_2) = 1$ in $N(u)$, so $f$ is a PID-function.
The weight of $f$ is $|D|$ which can be as small as $\fd2(G)$, completing the proof.
\end{proof}
\noindent In order to exploit the previous theorem, we prove the following result regarding the structure of any PID-function $f$ witnessing $\pid(G) = 3$.
\begin{lemma}
\label{lem:pid3-3ones}
Any PID-function $f$ of a connected graph $G$ witnessing $\pid(G) = 3$ assigns a weight of 1 to exactly three vertices and does not assign a weight of 2 to any vertex.
\end{lemma}
\begin{proof}
Suppose this was not the case, i.e., that instead $f$ set $f(u) = 2$ and $f(v) = 1$ for some distinct $u, v \in V(G)$.
Now consider any $v' \in N(v)$ such that $v' \neq u$.
Because $f$ is a PID-function of weight 3, it must hold that $f(v') = 0$. 
But because $v'$ is adjacent to $v$ and $f(v) = 1$, the weights on the neighbors of $v'$ assigned by $f$ cannot sum to exactly 2, contradicting the fact that $f$ is a PID-function.
\end{proof}
\noindent We are now ready to prove the main result of the section.
\begin{theorem}
A connected graph $G$ with $\pid(G) > 2$ has $\pid(G) = 3$ if and only if $G$ has a 2-fair dominating set $D$ of size~3.
\end{theorem}
\begin{proof}
Suppose that $\pid(G) = 3$ which is witnessed by a PID-function $f$.
By Lemma~\ref{lem:pid3-3ones}, $f$ has picked three vertices, say $a$, $b$, and $c$ such that $f(a) = f(b) = f(c) = 1$ and labeled every other vertex~0.
We claim that $D = \{a,b,c\}$ is a 2-fair dominating set of size~3.
Indeed, every vertex with label~0 must be adjacent to exactly two vertices of $D$ since $f$ is a PID-function, so the claim follows.

For the other direction, construct a PID-function $f$ from a 2-fair dominating set $D$ such that $f(v) = 1$ for $v \in D$ and $f(u) = 0$ for $u \notin D$.
Clearly, as $D$ is a 2-fair dominating set, every $u$ is adjacent to exactly two vertices labeled~1.
Further, because $|D| = 3$, we have that $\pid(G) \leq 3$.
As $\pid(G) > 2$, we conclude that $\pid(G) = 3$.
\end{proof}
\noindent It is also possible to state the same result in a different way.
To do this, we observe the following.
\begin{proposition}
Let $G=(V,E)$ be a connected graph.
A subset $S \subseteq V$ of size $s$ is an $\ell$-fair dominating set in $G$ if and only if $S$ is an $(s-\ell)$-fair dominating set in $\overline{G}$.
\end{proposition}
\noindent A 1-fair dominating set is also known as a \emph{perfect dominating set} (see Fellows and Hoover~\cite{Fellows1991}).
\begin{corollary}
Let $G=(V,E)$ be a connected graph.
A subset $S \subseteq V$ of size three is a 2-fair dominating set in $G$ if and only if $S$ is a perfect dominating set in $\overline{G}$.
\end{corollary}
\noindent We can then restate our earlier theorem as follows.
\begin{theorem}
\label{thm:pid-three}
A graph $G$ with $\pid(G) > 2$ has $\pid(G) = 3$ if and only if $\overline{G}$ has a perfect dominating set of size~3.
\end{theorem}

Let us then proceed to determine the perfect Italian domination number of complete multipartite graphs.
\begin{lemma}
\label{lem:k33-pid}
For every two integers $n_1,n_2 \geq 3$, it holds that $\pid(K_{n_1,n_2}) = 4$.
\end{lemma}
\begin{proof}
Let $G = K_{n_1,n_2}$. By Proposition 5, we have that $\pid(G) \geq 3$.
Further, the complement $\overline{G}$ of $G$ is the disjoint union of two cliques $K_{n_1}$ and $K_{n_2}$, and so every perfect dominating set in $\overline{G}$ has size two. Thus, Theorem~16 implies that $\pid(G) \geq 4$.
A matching upper bound is given by a function which assigns a label of 2 to exactly one vertex in each partite set, while setting all remaining labels to~0.
This completes the proof.
\end{proof}

\begin{lemma}
\label{lem:tripartite}
For every three integers $n_1,n_2,n_3 \geq 3$, it holds that $\pid(K_{n_1,n_2,n_3}) = 3$.
\end{lemma}
\begin{proof}
Let us denote $G = K_{n_1,n_2,n_3}$.
By Proposition~\ref{prop:dom}, $\pid(G) \geq 3$.
To give a matching upper bound, it suffices to notice that $\overline{G}$ is a disjoint union of three cliques $K_{n_1}$, $K_{n_2}$, and $K_{n_3}$.
A perfect dominating set of size three in $\overline{G}$ is given by choosing exactly one vertex from each component. 
By Theorem~\ref{thm:pid-three}, we conclude that $\pid(G) = 3$.
\end{proof}

\begin{lemma}
\label{lem:larger}
Given $k > 3$ integers $n_1,n_2,\ldots,n_k \geq 3$, it holds that $\pid(K_{n_1,n_2,\ldots,n_k}) = n_1 + n_2 + \cdots + n_k$.
\end{lemma}
\begin{proof}
Let $G$ be the complete multipartite graph $K_{n_1,n_2,\ldots,n_k}$ of order $n = n_1 + n_2 + \cdots + n_k$.
For the sake of contradiction, assume that a PID-function $f$ of $G$ with weight less than~$n$ exists.
By the pigeonhole principle, there must exist a vertex $u$ in a set $V_i$ of the $k$-partition of $G$ for some $1 \leq i \leq k$ with $f(u) = 0$. As such, to satisfy the conditions of PID-functions, the labels in the neighborhood of $u$ must account for a total value of $2$. This can be done in three ways:\\

%
\noindent {\bf Case 1:}\ There is a neighbor $v$ of $u$ in some $V_j$ with $j \neq i$ such that $f(v) = 2$.

%
\noindent {\bf Case 2:}\ There are neighbors $v$ and $v'$ of $u$, both in some $V_j$ with $j \neq i$ such that $f(v) = f(v') = 1$.

%
\noindent {\bf Case 3:}\ There are neighbors $v$ and $v'$ of $u$ in $V_j$ and $V_m$, respectively, with $j$, $m$, and $i$ pairwise distinct, such that $f(v) = f(v') = 1$.\\


We let $S_u = \{v \in N(u) \mid f(v) \neq 0 \}$ and observe that $S_u$ always includes vertices from no more than two partite sets. Clearly, every vertex in $V \setminus (V_i \cup S_u)$ has to be labeled 0. Since $k > 3$ this includes some partitite set $V_h$, with $i \neq h$ and $V_h \cap S_u = \emptyset$. Since the vertices in this set are labeled 0 and adjacent to the vertices of $S_u$, every vertex in $V_i$ must also be labeled $0$. In other words, the only vertices with positive labels are the ones in $S_u$, meaning that $f$ has a weight of $2$. By Proposition~\ref{prop:dom}, this is a contradiction, completing the proof.

\end{proof}
\noindent The previous lemmas together prove the following.
\begin{theorem}
\label{thm:compl-multi}
Let $G = K_{n_1,n_2,\ldots,n_k}$ be the complete $k$-partite graph on $n = n_1 + n_2 + \cdots + n_k$ vertices, where $n_i \geq 3$ for each $1 \leq i \leq k$. Then
\[
  \pid(G) = \left.
  \begin{cases}
    4, & \text{if } k = 2, \\
    3, & \text{if } k = 3, \\
    n, & \text{if } k \geq 4.
  \end{cases}\right.
\]
\end{theorem}

\begin{remark}
The complete multipartite graph $G = K_{n_1,n_2,\ldots,n_k}$ for $k \geq 4$ shows that the difference between $\romtwo(G)$ and $\pid(G)$ can be made arbitrarily large. Indeed, by Theorem~\ref{thm:compl-multi} we have that $\pid(G)$ is equal to the order of $G$, but $\romtwo(G) = 3$ as witnessed by labeling exactly one vertex~1 from three different sets of the $k$-partition of $G$ and labeling the remaining vertices~0. 
\end{remark}

\begin{remark}
Let $G = K_{n_1,n_2,n_3}$ be a complete tripartite graph with $n_i \geq 3$ for $1 \leq i \leq 3$. By Lemma~\ref{lem:tripartite}, $\pid(G) = 3$ while $\rom(G) = 4$ (see~\cite[Proposition~8]{Cockayne2004}).
Thus, it is not true that $\rom(G) \leq \pid(G)$ in general (cf.\ Proposition~\ref{prop:pid-bounds}).
\end{remark}

\section{On upper bounds for restricted graph classes}
\label{sec:upper-bounds}

Haynes and Henning~\cite{Haynes2019} proposed the problem of determining the best possible constant $c_\mathcal{G}$ such that $\pid(G) \leq c_\mathcal{G} \cdot n$ for all $n$-vertex graphs $G$ belonging to a particular class $\mathcal{G}$ of graphs.
In particular, they showed that if $\mathcal{G}$ is the class of connected bipartite graphs, then $c_\mathcal{G} = 1$, whereas if $\mathcal{G}$ is the class of trees (on at least 3 vertices), then $c_\mathcal{G} = 4/5$.
Further, the authors suggested to study the problem further when $\mathcal{G}$ would be e.g., the class of planar graphs or regular graphs.

In the following subsections, we settle precisely the question when $\mathcal{G}$ is the class of connected planar graphs by proving, perhaps surprisingly, that $c_\mathcal{G} = 1$.
In addition, we also completely settle the question when $\mathcal{G}$ is the class of connected cubic graphs by proving that $c_\mathcal{G} = 2/3$.
Further, when $\mathcal{G}$ is the class of $k$-regular graphs for $k \geq 5$, we show that $c_\mathcal{G} = 1$.
We conclude by observing that $c_\mathcal{G} = 1$ when $\mathcal{G}$ is the class of connected split graphs, implying that $c_\mathcal{G} = 1$ also when $\mathcal{G}$ is any superclass of split graphs, like the class of chordal graphs or more generally, the perfect graphs.

\subsection{Planar graphs}
In this subsection, we describe an infinite family of connected planar graphs $G$ that have $\pid(G)$ equal to their order, thus proving that $c_\mathcal{G} = 1$ when $\mathcal{G}$ is the class of connected planar graphs.

Let $J_1$ be the connected 10-vertex planar graph that is formed by adding two dominating vertices to $2K_2$ and then finishing by connecting a pendant vertex to every vertex except for two vertices of degree three (see Figure~\ref{fig:jewel}).
In particular, name the four support vertices of $J_1$ (i.e., those with a pendant in their neighborhood) so that $u$ and $v$ are those with degree five, and $x$ and $y$ are those with degree four.
The graph $J_2$ is obtained via \emph{widening} $J_1$ by connecting both $u$ and $v$ with the pendants of $x$ and $y$, say $x'$ and $y'$, respectively, and by introducing a new pendant vertex to each of $x'$ and $y'$.
The widening of $J_1$ to obtain $J_2$ is illustrated in Figure~\ref{fig:jewel}.
In total, a widening operation adds two vertices and six edges.
In general, the graph $J_\ell$ for any $\ell \geq 3$ is obtained recursively by widening $J_{\ell-1}$, which in turn is obtained by widening $J_{\ell-2}$, and so on.
Our goal is to show that $\pid(J_\ell) = n$, where $n$ is the order of $J_\ell$.
To this end, we make the following claims concerning any PID-function with weight less than~$n$.

\begin{figure}[t]
    \centering
        \includegraphics[width=0.75\textwidth,keepaspectratio]{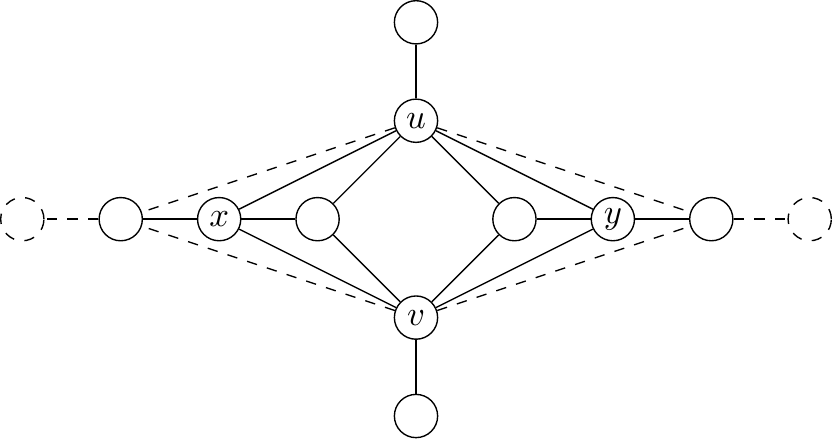} 
    \caption{The graph $J_1$ represented by solid vertices and edges. The graph $J_2$ is obtained by widening $J_1$, i.e., by adding the dashed elements to $J_1$.}
    \label{fig:jewel}
\end{figure}

\begin{lemma}
\label{lem:uv-sum-two}
Let $f$ be a PID-function of $J_\ell$ with weight less than $n$. 
It must hold for the support vertices $u$ and $v$ that $f(u) + f(v) \leq 2$.
\end{lemma}
\begin{proof}
If the latter was not the case, i.e., if $f(u) + f(v) > 2$, then none of the remaining non-pendant vertices could be labeled~0 because $u$ and $v$ are in the neighborhood of each such vertex. Importantly, this includes all the non-support vertices, which must therefore have an average weight of at least 1.
Further, if $f$ labels any pendant with 0, it must also label its support with 2. Thus, every pair comprising a pendant vertex and its support vertex will always contribute a weight of at least 2, again implying an average weight of at least 1 between them.
Combining the results on non-support vertices and pendant-support pairs, we arrive at a PID-function with weight at least $n$, proving the contrapositive of our lemma.
\end{proof}

\begin{lemma}
\label{lem:uv-nonzero}
Let $f$ be a PID-function of $J_\ell$.
The function $f$ must label $f(u) \neq 0$ and $f(v) \neq 0$.
\end{lemma}
\begin{proof}
For the sake of contradiction, suppose that $f(u) = 0$.
Because $f$ is a PID-function, it holds that $\sum_{u' \in N(u)} f(u') = 2$.
Clearly, the pendant of $u$ cannot be labeled~0, so first suppose that pendant of $u$ was labeled~2.
It follows that every other vertex adjacent to $u$ must be labeled~0.
But now it must be the case that $f(v) = 2$ and $f(y') \neq 0$, but $\{v, y'\} \subseteq N(y)$ with $f(y) = 0$, contradicting the fact that $f$ is a PID-function.
So it must be the case that the pendant of $u$ is labeled~1. 
It follows that precisely one other neighbor $a$ of $u$ is labeled~1 while the rest are labeled~0. Given that this also includes at least one non-support vertex and its neighbors which are not $u$ or $v$, this means that once again $f(v) = 2$.
But now a neighbor of $a$, labeled~0, is adjacent to $a$ (with label~1) and $v$ (with label~2), contradicting the fact that $f$ is a PID-function.
We conclude that $f(u) \neq 0$.
By a symmetric argument, $f(v) \neq 0$ under any valid PID-function $f$ (including those whose weight is less than~$n$, if any).
\end{proof}

\begin{lemma}
For any integer $\ell \geq 1$, it holds that $\pid(J_\ell) = n$.
\end{lemma}
\begin{proof}
For the sake of contradiction, suppose that there is a PID-function $f$ for $J_\ell$ for any $\ell \geq 1$ with weight less than~$n$.
By combining Lemma~\ref{lem:uv-sum-two} with Lemma~\ref{lem:uv-nonzero}, we know that any such $f$ must label $f(u) = f(v) = 1$.
Consider any vertex $a$ that is a common neighbor of both $u$ and $v$.
If $f(a) = 0$, all neighbors of $a$ must also be labeled~0.
In particular, we have $f(x) = 0$ or $f(y) = 0$. 
Without loss of generality, suppose that $f(x) = 0$, and observe that the pendant vertex $x'$ cannot receive any of the labels 0, 1, or 2 without violating the fact that $f$ is a valid PID-function, a contradiction.
Otherwise, if there is no such $a$ with $f(a) = 0$, the weight of $f$ is at least $n-4$ with only the pendants unlabeled.
Clearly, the two pendants of $u$ and $v$ cannot be labeled~0, but can be labeled~1.
For the pendants $x'$ and $y'$ of $x$ and $y$ there are two choices: either set (i) $f(x') = 0$ and $f(x) = 2$ or set (ii) $f(x') = f(x) = 1$, and similarly the same for $y$ and $y'$.
In both cases $f$ has weight $n$, a contradiction.
\end{proof}
\noindent The previous lemma establishes the main result of this subsection.
\begin{theorem}
\label{thm:planar-pid-n}
There is an infinite family of $n$-vertex connected planar graphs $G$ such that $\pid(G) = n$.
\end{theorem}
\noindent As a side remark, we can also see that for any $\ell \geq 1$, the \emph{treewidth} of $J_\ell$ is three. Thus, unlike for e.g., chromatic number, it is not true that the perfect Italian domination number of a graph could be bounded as a function of treewidth.

\subsection{Regular graphs}
\label{subs:regular}
In this subsection, we shift our focus to regular graphs.
As a main result here, we derive tight upper and lower bounds for the perfect Italian domination number of cubic graphs.

A \emph{strong matching}, also known as an \emph{induced matching}, is a set $M$ of edges of a graph $G$ such that no two edges in $M$ are connected by an edge of $G$.
Viewed differently, an induced matching is an independent set in the square of the line graph $G$.
The \emph{strong matching number}, denoted by $\im(G)$, is the size of a maximum induced matching of $G$.
For the next lemma, the key observation is that if $M$ is a strong matching in a cubic graph $G$, then $V(G) \setminus V(M)$ is a 2-fair dominating set of $G$.
\begin{lemma}
\label{lem:strong-matching-pid}
Every cubic graph $G$ with $n$ vertices has $\pid(G) \leq n - 2 \im(G)$.
\end{lemma}
\begin{proof}
Let $M$ be any strong matching of $G$.
Construct a vertex-labeling $f : V(G) \to \{0,1,2\}$ such that $f(u) = f(v) = 0$ for every $\{u,v\} \in M$ and label all other vertices~1.
Clearly, $f$ is a PID-function since every vertex $v$ with $f(v) = 0$ has two neighbors labeled~1 and one labeled~0.
The weight of $f$ is $n - 2 |M|$, which is equal to $n - 2 \im(G)$ when $|M| = \im(G)$.
\end{proof}
\noindent The following bound for the strong matching number will be useful for us.
\begin{theorem}[Joos~{et al.}~\cite{Joos2014}]
\label{thm:strong-matching-bound}
A cubic graph $G$ with $m$ edges has $\im(G) \geq m/9$.
\end{theorem}

\noindent Before proceeding, we mention that Chellali~{et al.}~\cite[Theorem~11]{Chellali2016} proved that $\romtwo(G) \geq 2n/(\Delta + 2)$, where $G$ is a connected $n$-vertex graph with maximum degree $\Delta$.
Combined with Proposition~\ref{prop:pid-bounds}, we obtain the following.
\begin{theorem}
\label{thm:pid-lower-bound-deg}
A connected graph $G$ on $n$ vertices with maximum degree $\Delta$ has $\pid(G) \geq {2n/(\Delta+2)}$.
\end{theorem}
\noindent We are now ready to establish the main result of this subsection.
\begin{theorem}
\label{thm:pid-cubic-tight}
Every connected cubic graph $G$ with $n$ vertices has $\tfrac{2}{5}n \leq \pid(G) \leq \tfrac{2}{3}n$. Moreover, these bounds are tight.
\end{theorem}
\begin{proof}
The lower bound follows from Theorem~\ref{thm:pid-lower-bound-deg} by having $\Delta = 3$.
The claimed upper bound follows by applying Lemma~\ref{lem:strong-matching-pid} for which we combine the fact that every cubic graph $G$ with $n$ vertices has $\tfrac{3}{2} n$ edges with Theorem~\ref{thm:strong-matching-bound}. That is, we see that
\begin{equation*}
\pid(G) \leq n - 2 \im(G) \leq n - 2(m/9) = n - 2(n/6) = 2n/3.
\end{equation*}

To see that the lower bound is tight, one can consider any connected cubic graph with 8 vertices. For instance, when $G$ is $Q_3$ (i.e., the Cartesian product of three 2-vertex paths $P_2 \Box P_2 \Box P_2$), we have that $\pid(G) = 4 = \lceil 16/5 \rceil$.
To see that the upper bound is tight, one can consider $G$ defined as the Cartesian product of $K_3$ and $K_2$.
Clearly, $G$ does not satisfy the condition of Proposition~\ref{prop:dom}.
Further, $\overline{G}$ is isomorphic to the 6-cycle, which does not admit a perfect dominating set of size three, so by Theorem~\ref{thm:pid-three} it holds that $\pid(G) \geq 4$.
By our upper bound $\pid(G) \leq 4$ as well, so both bounds are tight.
\end{proof}
\noindent Another example to see that $\pid(G) \leq \tfrac{2}{3} n$ is tight is $K_{3,3}$ which by Lemma~\ref{lem:k33-pid} has perfect Italian domination number equal to four.

At this point, it is interesting to contrast the upper bound of the previous theorem for cubic graphs to the result of Theorem~\ref{thm:compl-multi} which implies that there are regular graphs that do not admit a PID-function of weight less than their order~$n$. 
So more precisely, for what values of $k$ do there exist $k$-regular graphs that do not have PID-functions of weight less than $n$? In what follows, we show that there is an infinite family of $k$-regular graphs $G$ for each $k \geq 5$ such that $\pid(G) = n$.

\begin{figure}[t]
    \centering
        \includegraphics[width=0.55\textwidth,keepaspectratio]{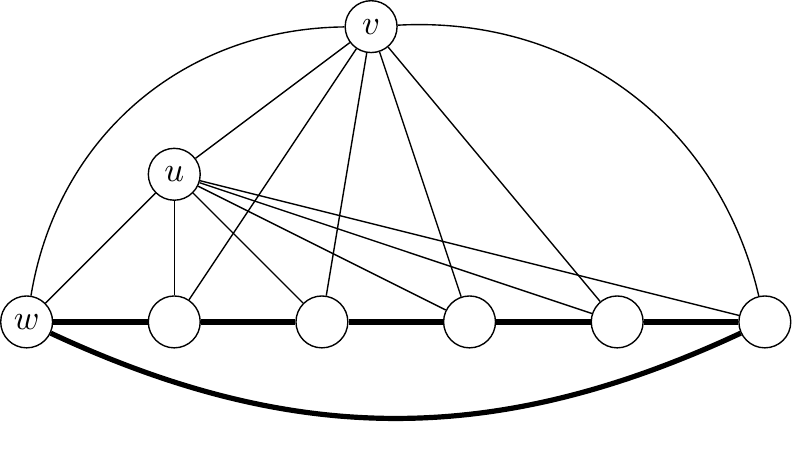} 
    \caption{The graph $\KC(a,r,b,s)$ for $a > 2$, $r=2$, $b=6$ and $s=1$. Other vertices in the partite sets containing $v$ and $u$ are omitted along with edges originating from them. The partite set of size $b$ containing $w$ contains a cycle shown with thick edges.}
    \label{fig:kc}
\end{figure}

We begin by introducing a construction for handling the case when $k \geq 9$.
Fix four non-negative integers $a$, $b$, $r$, and $s$ such that $2 < a < b$.
Let $\KC(a, r, b, s)$ be the graph obtained by starting from $K_{a_1,a_2,\ldots,a_r,b_1,b_2,\ldots,b_s}$ with $a_i = a$ and $b_j = b$ for all $1 \leq i \leq r$ and $1 \leq j \leq s$ and by adding $b$ edges to each $b$-sized partite set to connect its vertices arbitrarily into a cycle.
For an illustration of the definition, see Figure~\ref{fig:kc}.

\begin{lemma}
\label{lem:kcarbs}
Let $a$, $b$, $r$, and $s$ be non-negative integers such that $b > 4$, $r > 1$ and $2 < a < b$.
It holds that $G = \KC(a,r,b,s)$ has $\pid(G) = ar+bs = n$.
\end{lemma}
\begin{proof}
We denote the partite sets by $U_i$ when they are of size~$a$ and $V_j$ when they are of size~$b$.
For the purposes of contradiction, we will assume that there exists a PID-function which assigns the label $0$ to some vertex $u$. As a first step, we will also assume that such a vertex is in a set of size $a$, say (without loss of generality) $U_1$. The neighbourhood of $u$ consists of every vertex not in $U_1$, which includes at least $U_2$ and $V_1$. This neighbourhood must also account for labels summing up to~$2$. We consider all the possible cases in the following.\\

\noindent {\bf Case 1:} $u$ has non-$0$ neighbors in a single $a$-sized partite set, say (without loss of generality)~$U_2$. This can either be done with one vertex labeled $2$ or two vertices labeled $1$. At the same time, all the remaining vertices outside $U_1$ (including the vertices of $V_1$), are also forced to have a label of $0$, as they neighbor $u$. In turn, since these are adjacent to the vertices in $U_2$, they force the remaining vertices in $U_1$ to also be labeled $0$. To summarise, every vertex is labeled $0$, except for either one or two vertices in $U_2$ accounting for labels totalling $2$. Since $a > 2$, this implies that there is at least one vertex in $U_2$ labeled $0$. However, all its neighbors are also labeled $0$, creating a contradiction.\\

\noindent {\bf Case 2:} $u$ has non-$0$ neighbors in a single $b$-sized partite set, say (without loss of generality) $V_1$. As before, this can be done with one or two vertices. By the same argument, we have all the vertices outside $V_1$ labeled $0$. Since this is a PID-function, its restriction to $V_1$ must also be a PID-function, since no vertex in $V_1$ has non-$0$ neighbors outside $V_1$. But, given that $b > 4$, we need a PID-function on a cycle of length $5$ or greater, with total weight equal to $2$, which is impossible.\\

\noindent {\bf Case 3:} $u$ has non-$0$ neighbors in two distinct partite sets. Clearly, in this case, there must be two vertices $v_1$ and $v_2$ labeled~$1$. Let these vertices be in the partite sets $X$ and $Y$. If there exists some partite set other than $U_1$, $X$, and $Y$, this forces all remaining vertices in the entire graph to be labeled $0$. Thereby, the PID-function must have a total weight of $2$. But, by the construction of $\KC(a,r,b,s)$, there is no pair of vertices, each of which dominate the entire graph, causing a contradiction. This leaves the subcase where $r = 2$ and $s = 1$, with $v_1 \in V_1$ and $v_2 \in U_2$. Consider the vertices labeled with~0 in $V_1$. Since $b > 4$, at least two of them are non-adjacent to $v_1$. To have a PID-function, these must have some neighbor labeled~$1$. This neighbor must necessarily be in $U_1$. But then, consider the neighbors of $v_1$ within $V_1$. These have a neighborhood of weight~$3$, creating yet another contradiction.\\

From all the above cases we deduce that no vertex in a $U_i$ can be labeled $0$. Then, we must have some vertex in a $V_j$ labeled $0$. But then, it has at least $2a > 4$ neighbors in the $U_i$ partite sets, labeled $1$ or $2$, creating our final contradiction and proving that there is no PID-function with weight less than $n$.
\end{proof}

Following Haynes and Henning~\cite{Haynes2019}, recall that for a given class of graphs $\mathcal{G}$, we are interested in determining the best possible constant $c_\mathcal{G}$ such that $\pid(G) \leq c_\mathcal{G} \cdot n$ for all graphs of order $n$ when $G$ is a member of $\mathcal{G}$.
With this convenient notation at hand, we state the following.
\begin{theorem}
\label{thm:reg-pid-n}
For each $k \geq 5$, there is an infinite family of $k$-regular graphs $\mathcal{G}_k$ such that $c_{\mathcal{G}_k} = 1$.
\end{theorem}
\begin{proof}
To prove our claim, it will suffice by Proposition~\ref{prop:disconnected} to demonstrate the existence of a $k$-regular graph $G$ for which $\pid(G) = |V(G)|$ for each $k \geq 5$.
Indeed, one can take multiple disjoint copies of such $G$ to obtain an infinite family of $k$-regular graphs that do not admit a PID-function of weight less than $|V(G)|$.

Let us first consider the case of $k \geq 9$. We have the following subcases:\\

\noindent {\bf Case 1:} $k$ is a prime number of the form $3t + 2$, with $t \geq 3$. We consider $\KC(t,3,t+2,1)$.\\

\noindent {\bf Case 2:} $k$ is a prime number of the form $3t + 1$, with $t > 3$. We consider $\KC(t-1,2,t+1,2)$.\\

\noindent {\bf Case 3:} $k$ is a composite number of the form $2p$ with $p > 4$ prime. We consider $\KC(p-1,2,p+1,1)$.\\

\noindent {\bf Case 4:} $k$ is a composite number of the form $ab$ where $a, b > 2$. We consider the complete $(a + 1)$-partite graph $K_{b, b, \ldots, b}$.\\

First, we observe that these cases account for every possible $k \geq 9$. Furthermore, each of the given graphs is $k$-regular. Finally, the graphs in the first three cases meet the conditions of Lemma~\ref{lem:kcarbs}, while the graph in the fourth case meets the conditions of Lemma~\ref{lem:larger}; thus they all have $\pid$ equal to their order.

Finally, for $5 \leq k \leq 8$ we turn to a computer search with the help of House of Graphs~\cite{Brinkmann2013}, an online database for ``interesting'' graphs, and the \texttt{genreg} program of Meringer~\cite{Meringer1999}. For full details, we refer the curious reader to the appendix.
\end{proof}
\noindent Despite some effort, we were unable to discover a 4-regular graph $G$ with $\pid(G) = n$. We leave this (or perhaps its negation) as an open question.

\subsection{Split graphs}
In this subsection, we consider \emph{split graphs} defined as graphs whose vertex set can be partitioned into a clique and an independent set.
Split graphs are highly restricted graphs forming a subclass of chordal graphs, which in turn are a subclass of perfect graphs.

For any $\ell \geq 6$, let $S_\ell$ be the split graph of order $\ell + 2$ obtained by starting from a $K_\ell$ and by choosing four distinct arbitrary vertices $\{a,b,c,d\}$ of it and adding two new vertices $x$ and $y$ with the edges $\{ xa, xb, xc \} \cup \{ yd \}$ (see Figure~\ref{fig:split-fam}).
That is, $I = \{ x, y \}$ forms an independent set, while $K = V(S_\ell) \setminus I$ induces a (unique) clique of size $\ell$.

\begin{figure}[t]
    \centering
        \includegraphics[width=0.55\textwidth,keepaspectratio]{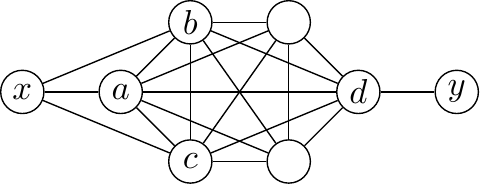} 
    \caption{The parametric split graph $S_\ell$ for $\ell = 6$.}
    \label{fig:split-fam}
\end{figure}

\begin{lemma}
For any $\ell \geq 6$, it holds that $\pid(S_\ell) = \ell + 2 = n$.
\end{lemma}
\begin{proof}
For the sake of contradiction, suppose that $\pid(S_\ell) < n$ and that this is witnessed by a PID-function $f$.
Because $f$ has weight less than~$n$, there must exist at least one vertex $v$ such that $f(v) = 0$.
Suppose that $f(x) = 0$.
Then, without loss of generality, there are two possibilities: either (i) $f(a) = 2$ and $f(b) = f(c) = 0$ or (ii) $f(a) = 0$ and $f(b) = f(c) = 1$.
In both cases, it follows that all the other vertices of $K$ must be labeled~0 by $f$.
In particular, it holds that $f(d) = 0$, but now there is no label $f$ can assign to $y$.
Thus, $f(x) \neq 0$. 

Without loss of generality, suppose that $f(a) = 0$.
Now, if $f(x) = 2$, it must be that $f(b) = f(c) = 0$.
Again, by the same argument as above, there is no label $f$ can assign to~$y$.
Thus, if $f(a) = 0$ then $f(x) = 1$ must hold.
Now, $f$ must label exactly one vertex of the $\ell - 1$ vertices of $K$ with~1 and the other with~0.
But then there is always at least one vertex in $u \in K \setminus \{ b,c \}$, which is distinct from $d$ as $\ell \geq 6$, such that $\sum_{u' \in N(u)} f(u') = 1$, contradicting the fact that $f$ is a PID-function.

Because none of $a$, $b$, and $c$ can be labeled~0 by $f$, it follows that $f(a) + f(b) + f(c) \geq 3$, and thus $f(u) \neq 0$ for every $u \in K$.
At this point, the only possibility is that $f(y) = 0$. 
It follows that $f(d) = 2$.
As no other vertex can be labeled~0, we can label every remaining vertex~1.
But now the weight of $f$ is $n$, a contradiction.
We conclude that $\pid(S_\ell) = n$, which is what we wanted to prove.
\end{proof}
\noindent The previous lemma establishes the following result.
\begin{theorem}
There is an infinite family of connected split graphs $\mathcal{G}$ such that $c_\mathcal{G} = 1$.
\end{theorem}
\noindent We can further contrast this result with the fact that threshold graphs, which are precisely the $P_4$-free split graphs, always admit a PID-function of weight at most~2 by Proposition~\ref{prop:pid-threshold}.

\section{Hardness of perfect Italian domination}
\label{sec:hardness}

In this section, we prove that perfect Italian domination is $\NP$-complete, even when restricted to bipartite planar graphs.
In all our hardness proofs, we omit explicitly showing membership to $\NP$ as it is an easy exercise.

To prove the claimed result, we give a polynomial-time reduction from \probPlanarXC in which we are given a finite set $X$ with $|X| = 3q$ and a family $\mathcal{C}$ of 3-element subsets of $X$.
The goal is to decide whether there is a subfamily $\mathcal{C}'$ of $\mathcal{C}$ such that every element of $X$ appears in exactly one element of $\mathcal{C}'$.
Every instance $(X,\mathcal{C})$ is associated with a bipartite \emph{incidence graph}, in which the first set of the bipartition corresponds to elements in $X$ and the second to elements in $\mathcal{C}$.
The edge set is defined such that two vertices are connected precisely when an element of $X$ is contained in an element of $\mathcal{C}$.
In \probPlanarXC, we have the further constraint the incidence graph is both bipartite \emph{and} planar.
This problem was shown to be $\NP$-complete by Dyer and Frieze~\cite{Dyer1986}.
\begin{theorem}[Dyer and Frieze~\cite{Dyer1986}]
\label{thm:planar-xc-npc}
\probPlanarXC is $\NP$-complete.
\end{theorem}

Before describing our reduction, let us introduce the following gadget.
For any positive integer $\ell \geq 1$, the \emph{fish gadget} $F_\ell$ is constructed by starting from the disjoint union of $2 \ell$ vertices partitioned into two equally-sized sets $T$ and $M$, and by adding two vertices $x$ and $y$ such that $y$ is adjacent to every vertex in $T \cup M$ and $x$ is adjacent to every vertex in $M$. 
Thus, $F_\ell$ has a total of $2 \ell + 2$ vertices, with $\ell$ vertices of degree two and $\ell$ vertices of degree one.
The fish gadget is illustrated in Figure~\ref{fig:fish-gadget}.

\begin{proposition}
\label{prop:fish-large}
For any $\ell \geq 3$, any PID-function $f$ of $F_\ell$ has weight at least $\ell+2$ if $f(x) = 1$. Similarly, if $f(x) = 2$, $f$ has weight at least $\ell+4$.
\end{proposition}

\begin{figure}[t]
    \centering
        \includegraphics[width=0.35\textwidth,keepaspectratio]{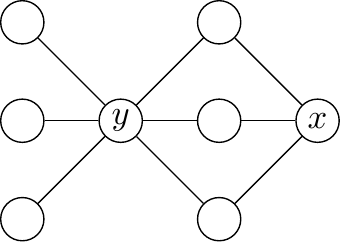} 
    \caption{A fish gadget $F_\ell$ for $\ell = 3$.}
    \label{fig:fish-gadget}
\end{figure}

\noindent We say that a vertex $v$ for which $f(v) = 0$ is \emph{satisfied} if $\sum_{u \in N(v)} f(u) = 2$.
Even more precisely, we say that such a $v$ is \emph{out-satisfied} (with respect to some subgraph $H$ of $G$) if $\sum_{u \in N(v) \setminus V(H)} f(u) = 2$.
Similarly, $v$ is \emph{in-satisfied} if $\sum_{u \in N(v) \wedge u \in V(H)} f(u) = 2$.
For the following statement, the subgraph $H$ is to be understood to be the gadget $F_\ell$ itself. Moreover, $x$ is the only vertex that will be connected to vertices outside of $F_\ell$ in a larger construction to follow.
\begin{proposition}
\label{prop:fish-outsat}
Let $\ell \geq 3$ and $f$ be some PID-function on a graph $G$ containing a copy of $F_\ell$ as a subgraph, with $x$ a cut-vertex. If $f(x) = 0$, then the restriction of $f$ on the copy of $F_{\ell}$ has optimal weight~2 if $x$ is out-satisfied, optimal weight~4 if $x$ is in-satisfied, and optimal weight~3 if $x$ is neither out-satisfied nor in-satisfied. 
\end{proposition}
\begin{proof}
In the first case, set $f(y) = 2$ and label other vertices 0.
In the second case, set $f(y) = 2$, label an arbitrary vertex in $M$ with 2, and label other vertices 0.
In the third case, set $f(y) = 2$, label an arbitrary vertex in $M$ with 1, and label other vertices 0.
It is easy to see that these labelings are optimal.
\end{proof}
Let us call \probPID the problem where we are given a graph $G$ and an integer $k$, and the goal is to decide whether $G$ admits a PID-function of weight at most~$k$.
\begin{theorem}
\label{thm:pid-npc}
\probPID is $\NP$-complete for bipartite planar graphs.
\end{theorem}
\begin{proof}
Let $(X,\mathcal{C})$ be an instance of \probPlanarXC, such that $X = \{1,2,\ldots,n\}$, $\mathcal{C} = \{ C_1, C_2, \ldots, C_t \}$, and $|X| = 3q$.
We proceed by describing a polynomial-time reduction to \probPID as follows.

Let $H$ be the bipartite incidence graph of $(X,\mathcal{C})$, which we can also safely assume to be planar by Theorem~\ref{thm:planar-xc-npc}.
So more precisely, $V(H) = \{x_1,x_2,\ldots,x_n\} \cup \{ c_1, c_2, \ldots, c_t \}$ with $x_i$ and $c_j$ adjacent precisely when $i$ is a member of $C_j \in \mathcal{C}$.
Let $k = 6q+2t$.
To obtain $G$ from $H$, identify $x_i$ for $i \in [n]$ with a fish gadget $F_k$ (at its vertex $x$) and attach to $c_j$ for $j \in [t]$ two pendants $c'_j$ and $c''_j$.
We tacitly name $y_i$ the vertex $y$ of a fish gadget corresponding to the vertex $x_i$.
Clearly, because the fish gadget is both bipartite and planar, $G$ is bipartite and planar as well.
We claim that $(X,\mathcal{C})$ has an exact cover if and only if $G$ admits a PID-function of weight at most~$k$.

Let $\mathcal{C'}$ be an exact cover of $(X,\mathcal{C})$.
We construct a vertex-labeling $f$ of $G$ such that $f(c_j) = 2$ for $C_j \in \mathcal{C'}$; all other vertices $c_j$ not in $\mathcal{C'}$ are labeled~0.
Here, if $f(c_j) = 2$, we set $f(c'_j) = f(c''_j) = 0$ and if $f(c_j) = 0$, then $f(c'_j) = f(c''_j) = 1$.
At this point, the labels used by $f$ have weight $2t$.
For each $i \in [n]$, we label $f(x_i) = 0$, $f(y_i) = 2$, and all the remaining vertices~0.
As $|X| = 3q$ and $f(y_i) = 2$, the weight of $f$ is exactly $2 \cdot 3q + 2t = 6q+2t = k$.
Now, since $\mathcal{C'}$ is an exact cover, every $x_i$ is out-satisfied by some $c_j$ corresponding to a $C_j \in \mathcal{C'}$.
For each $j \in [t]$, if $f(c_j) = 0$, then $c_j$ is satisfied by $f(c'_j) + f(c''_j) = 2$.
It follows that $f$ is a PID-function.

Conversely, suppose that $f$ is a PID-function of weight $k$.
It holds for every $i \in [n]$ that $f(x_i) = 0$ for otherwise $f$ would have weight at least $k+2 > k$ by Proposition~\ref{prop:fish-large}.
Further, as $S = \bigcup_{j \in [t]} \{c_j,c'_j,c''_j\}$ requires labels of weight at least $2t$, it follows by Proposition~\ref{prop:fish-outsat} that each $x_i$ must be out-satisfied for otherwise $f$ would have weight at least $2t+6q-2+3 = k + 1 > k$.
It follows that $f$ has allocated labels of weight $k - 6q = 2t$ to~$S$.
Further, this is only possible if $f(c_j) \neq 1$ for $j \in [t]$ for otherwise $f$ would have weight at least $6q + 2t - 2 + 3 = k + 1 > k$, to properly label $\{c_j,c'_j,c''_j\}$.
Therefore, since $f$ is a PID-function, every $x_i$ is out-satisfied by exactly one $c_j$ for which $f(c_j) = 2$.
Consequently, $\mathcal{C'} = \{ C_j \mid f(c_j) = 2 \}$ is an exact cover of $(X,\mathcal{C})$.
\end{proof}

It is worth mentioning that the earlier result of Chellali~{et al.}~\cite[Theorem~18]{Chellali2016} regarding the hardness of computing $\romtwo(G)$ also works for bipartite planar graphs.
Let us call \probRomanTwoD the problem of deciding whether given a graph $G$ and an integer $k$, it is true that $\romtwo(G) \leq k$.
\begin{theorem}
\probRomanTwoD is $\NP$-complete for bipartite planar graphs.
\end{theorem}
\begin{proof}
Chellali~{et al.}~\cite[Theorem~18]{Chellali2016} prove $\NP$-completeness of \probRomanTwoD for bipartite graphs by a polynomial-time reduction from an arbitrary instance $(X,\mathcal{C})$ of \probXC.
In short, their reduction begins from the bipartite incidence graph $H$ of $(X,\mathcal{C})$, but replaces every vertex corresponding to a $C \in \mathcal{C}$ with a $C_6$ with a chord followed by a 2-vertex path.
Because this gadget is both bipartite and planar, we ensure that the instance $G$ of \probRomanTwoD is both bipartite and planar by assuming that $H$ is planar.
By Theorem~\ref{thm:planar-xc-npc}, we can do this safely, so the result follows.
\end{proof}

\section{Open problems}
\label{sec:problems}
In this section, we conclude by highlighting some open problems arising from our work.

We begin with the following complexity-theoretic statement.
\begin{conjecture}
For every $k \geq 3$, \probPID is $\NP$-complete for the class of $k$-regular graphs. 
\end{conjecture}

In the light of our construction in the proof of Theorem~\ref{thm:planar-pid-n}, it might be interesting to consider other planar graphs $G$ with $\pid(G) = n$.
We verified by a computer search the smallest planar graph $G$ with $\pid(G) = n$ to have $n=7$ vertices, and there are no other such planar graphs on 7 vertices.
Thus, one might ask the following.
\begin{problem}
Can we characterize the connected $n$-vertex planar graphs $G$ such that $\pid(G) = n$, or at least find some conditions for this to hold?
\end{problem}
\noindent Also, after Theorem~\ref{thm:planar-pid-n}, it is natural to raise the question of Haynes and Henning~\cite{Haynes2019} for the class of \emph{bipartite} planar graphs.
At the same time, given our $\NP$-completeness result Theorem~\ref{thm:pid-npc}, one should not expect a polynomial-time characterization for this class.
\begin{problem}
Determine the best possible constant $c_\mathcal{G}$ such that $\pid(G) \leq c_\mathcal{G} \times n$ for all $n$-vertex graphs $G$ belonging to the class of connected bipartite planar graphs $\mathcal{G}$.
\end{problem}
In light of Theorem~\ref{thm:reg-pid-n}, we find the same question interesting for 4-regular graphs.
Despite some effort, we were unable to find a 4-regular graph $G$ on $n$ vertices for which $\pid(G) = n$ and thus conjecture that an upper bound better than~$n$ does exist.
Further, if we insist on the family of graphs in Theorem~\ref{thm:reg-pid-n} to be connected, does the statement still hold?
In general, we find the further study of perfect Italian domination interesting for regular graphs.

\bibliographystyle{abbrv}
\bibliography{bibliography}

\begin{thebibliography}{10}

\bibitem{Brinkmann2013}
G.~Brinkmann, K.~Coolsaet, J.~Goedgebeur, and H.~Mélot.
\newblock {House of Graphs: A database of interesting graphs}.
\newblock {\em Discrete Appl. Math.}, 161(1):311--314, 2013.

\bibitem{Caro2012}
Y.~Caro, A.~Hansberg, and M.~Henning.
\newblock Fair domination in graphs.
\newblock {\em Discrete Math.}, 312(19):2905--2914, 2012.

\bibitem{Chaluvaraju2010}
B.~Chaluvaraju, M.~Chellali, and K.~Vidya.
\newblock Perfect k-domination in graphs.
\newblock {\em Australas. J. Combin.}, 48:175--184, 2010.

\bibitem{Chaluvaraju2014}
B.~Chaluvaraju and K.~Vidya.
\newblock Generalized perfect domination in graphs.
\newblock {\em J. Comb. Optim.}, 27(2):292--301, 2014.

\bibitem{Chambers2009}
E.~W. Chambers, B.~Kinnersley, N.~Prince, and D.~B. West.
\newblock {Extremal problems for Roman domination}.
\newblock {\em SIAM J. Discrete Math.}, 23(3):1575--1586, 2009.

\bibitem{Chellali2012}
M.~Chellali, O.~Favaron, A.~Hansberg, and L.~Volkmann.
\newblock {$k$-Domination and $k$-Independence in Graphs: A Survey}.
\newblock {\em Graphs Combin.}, 28(1):1--55, 2012.

\bibitem{Chellali2013}
M.~Chellali, T.~W. Haynes, S.~T. Hedetniemi, and A.~McRae.
\newblock [1, 2]-sets in graphs.
\newblock {\em Discrete Appl. Math.}, 161(18):2885--2893, 2013.

\bibitem{Chellali2016}
M.~Chellali, T.~W. Haynes, S.~T. Hedetniemi, and A.~A. McRae.
\newblock Roman $\{2\}$-domination.
\newblock {\em Discrete Appl. Math.}, 204:22--28, 2016.

\bibitem{Cockayne2004}
E.~J. Cockayne, P.~A. Dreyer, S.~M. Hedetniemi, and S.~T. Hedetniemi.
\newblock Roman domination in graphs.
\newblock {\em Discrete Math.}, 278(1):11--22, 2004.

\bibitem{Dyer1986}
M.~E. Dyer and A.~M. Frieze.
\newblock Planar {3DM} is {NP}-complete.
\newblock {\em J. Algorithms}, 7(2):174--184, 1986.

\bibitem{Fellows1991}
M.~R. Fellows and M.~N. Hoover.
\newblock Perfect domination.
\newblock {\em Australas. J. Combin.}, 3:141--150, 1991.

\bibitem{Gera2018}
R.~Gera, T.~W. Haynes, S.~T. Hedetniemi, and M.~A. Henning.
\newblock An annotated glossary of graph theory parameters, with conjectures.
\newblock In {\em Graph Theory}, pages 177--281. Springer, 2018.

\bibitem{Hansberg2013}
A.~Hansberg.
\newblock Reviewing some results on fair domination in graphs.
\newblock {\em Electron. Notes Discrete Math.}, 43:367--373, 2013.

\bibitem{Haynes2019}
T.~W. Haynes and M.~A. Henning.
\newblock {Perfect Italian domination in trees}.
\newblock {\em Discrete Appl. Math.}, 260:164--177, 2019.

\bibitem{Henning2017}
M.~A. Henning and W.~F. Klostermeyer.
\newblock Italian domination in trees.
\newblock {\em Discrete Appl. Math.}, 217:557--564, 2017.

\bibitem{Joos2014}
F.~Joos, D.~Rautenbach, and T.~Sasse.
\newblock Induced matchings in subcubic graphs.
\newblock {\em SIAM J. Discrete Math.}, 28(1):468--473, 2014.

\bibitem{Liedloff2008}
M.~Liedloff, T.~Kloks, J.~Liu, and S.-L. Peng.
\newblock {Efficient algorithms for Roman domination on some classes of
  graphs}.
\newblock {\em Discrete Appl. Math.}, 156(18):3400--3415, 2008.

\bibitem{Mahadev1995}
N.~V.~R. Mahadev and U.~N. Peled.
\newblock {\em Threshold graphs and related topics}, volume~56.
\newblock Elsevier, 1995.

\bibitem{Meringer1999}
M.~Meringer.
\newblock Fast generation of regular graphs and construction of cages.
\newblock {\em J. Graph Theory}, 30(2):137--146, 1999.

\bibitem{Revelle2000}
C.~S. ReVelle and K.~E. Rosing.
\newblock {Defendens Imperium Romanum: A Classical Problem in Military
  Strategy}.
\newblock {\em Amer. Math. Monthly}, 107(7):585--594, 2000.

\bibitem{Stewart1999}
I.~Stewart.
\newblock {Defend the Roman empire!}
\newblock {\em Sci. Am.}, 281(6):136--138, 1999.

\end{thebibliography}

\section*{Appendix}
As mentioned in the proof of Theorem~\ref{thm:reg-pid-n}, we describe here the $k$-regular graphs $G$ on $n$ vertices for which $\pid(G) = n$ for $5 \leq k \leq 8$. 
The mentioned graphs are listed in Table~\ref{tbl:reg-pid-n}.
Here, the first column stands for the internal identifier at House of Graphs (HoG) at the time of writing. 
The second column is the degree of the graph, with $n$ and $m$ the order and size of the graph, respectively.
The last column describes the graph in the well-known graph6 format.

We make no attempt at optimizing the order or size of the graphs.
However, for $k=5$, it can be noted that there is no example of smaller order as the only 5-regular graph with fewer than 8 vertices is $K_6$ which admits a PID-function of weight two.

\begin{table*}[t]
\caption{Examples of $k$-regular graphs with $n$ vertices and $m$ edges for which $\pid(G) = n$.}
\label{tbl:reg-pid-n}
\centering
\def\arraystretch{1}
\begin{tabular}{lllll}
\toprule 
\textbf{HoG ID} & $k$ & $n$ & $m$ & \textbf{graph6} \\ 
\midrule
34402 & 5 & 8 & 20 & \verb|G}qzp{| \\
\rowcol 33136 & 6 & 12 & 36 & \verb|KvyCJlmF_{kN| \\
28355 & 7 & 24 & 84 & \verb|WsaCC???Wg_qK@WBGQOVS@woL`aES@pHC[`a[CFBRW?Nq??| \\
\rowcol 32790 & 8 & 12 & 48 & \verb|K~~LnNwFy^e~| \\
\bottomrule 
\end{tabular}
\end{table*}

\end{document}